\documentclass[11pt]{article}
\usepackage{fullpage}
\usepackage{graphicx}
\usepackage{subfigure}
\usepackage{amsmath,amssymb,amsthm}
\usepackage{hyperref}

\newcommand{\remove}[1]{}
\input{psfig.sty}

\usepackage{algorithm}
\usepackage{algorithmic}

\newtheorem{lemma}{Lemma}

\newtheorem{observation}{Observation}

\usepackage{color}

\usepackage{color}

\title{Greedy is good: \\An experimental study on minimum clique cover and
maximum independent set problems for randomly generated rectangles}

\author{Ritankar Mandal\footnote{Indian Statistical Institute, Kolkata 700108,
India} \and Anirban Ghosh\footnote{Dept. of CS, University of
Wisconsin-Milwaukee, Milwaukee, WI 53201, USA} \and Sasanka
Roy\footnote{Chennai Mathematical Institute, Chennai,
India} \and Subhas C. Nandy$^*$}

\date{}

\begin{document}
\thispagestyle{empty}
\maketitle
\sloppy

\begin{abstract}
Given a set ${\cal R}=\{R_1,R_2, \ldots, R_n\}$ of $n$ randomly positioned 
axis parallel rectangles in 2D, the problem of computing the minimum clique
cover (MCC) and maximum independent set (MIS) for the intersection graph
$G({\cal R})$ of the members in $\cal R$ are both computationally hard
\cite{CC05}. For the MCC problem, it is proved that polynomial time constant
factor approximation is impossible to obtain \cite{PT11}. Though such a
result is not proved yet for the MIS problem, no polynomial time constant
factor approximation algorithm exists in the literature. We study the
performance of greedy algorithms for computing these two parameters of $G({\cal
R})$. Experimental results shows that for each of the MCC and MIS problems, the
corresponding greedy algorithm produces a solution that is very close to its
optimum solution. Scheinerman \cite{Scheinerman80} showed that
the size of MIS is tightly bounded by $\sqrt{n}$ for a random instance of the 1D
version of the problem, (i.e., for the interval graph). Our experiment shows
that the size of independent set and the clique cover produced by the greedy
algorithm is at least $2\sqrt{n}$ and at most $3\sqrt{n}$, respectively. Thus
the experimentally obtained approximation ratio of the greedy algorithm for MIS
problem is at most $\frac{3}{2}$ and the same for the MCC problem is at least
$\frac{2}{3}$. Finally we will provide refined greedy algorithms based on a
concept of {\it simplicial rectangle}. The characteristics of this algorithm may
be of interest in getting a provably constant factor approximation algorithm for
random instance of both the problems. We believe that the result also
holds true for any finite dimension.  
\end{abstract}

{\bf Keywords:} Minimum clique cover, maximum independent set, rectangle
intersection graph, approximation algorithm, empirical study
\section{Introduction}
Let $G({\cal R})$ be the intersection graph of a set ${\cal R}=\{R_1,R_2,\ldots,
R_n\}$ of $n$ randomly placed polygonal objects, e.g., rectangles, circles,
polygons, etc., in a 2D region. In the practical applications two types of
cliques in the geometric intersection graph are considered, namely (i) graphical
clique and (ii) geometric clique. A graphical clique $C$ in $G({\cal R})$ is a
maximal complete subgraph of $G({\cal R})$. A geometric clique $C'$ consists of
a maximal set of objects $C' \subseteq {\cal R}$ such that they have a common
point in their interior. Thus, a geometric clique is always a graphical clique,
however the converse is not true. If $\cal R$ is a set of polygonal objects,
then the problem of computing the minimum geometric clique cover is NP-hard
\cite{FowlerPT81}. However, for a set $\cal R$ of axis-parallel rectangles, a
graphical clique is always a geometric clique since the axis-parallel rectangles
satisfy Helly property. Thus, the minimum clique cover of the graph $G({\cal
R})$ is same as the minimum number of points required to stab all the rectangles
in $\hat{\cal R}$. It is easy to show that the number of cliques present in
$G({\cal R})$ may be $O(n^2)$ in the worst case \cite{NB}. But, if $\cal R$
is a set of $c$-oriented polygons ($c \geq 5$), the Helly property does not
hold. Nilson \cite{NX} proved that the number of geometric clique in $G({\cal
R})$ can be at most $\tau(2,c)\phi({\cal R})\log_2^{c-1}(\phi({\cal R})+1)$,
where $\tau(2,c)$ is the Gallai number of the pairwise intersecting $c$-oriented
polygons and Let $\phi({\cal R})$ denotes the packing number of $\cal R$, that
is the maximum number of pairwise disjoint objects in $\cal R$. The same paper
also provides a $t(n,c)+O(nc\log(\phi(\cal R)))$ time algorithm for computing
the minimum geometric clique cover of $G({\cal R})$, where $t(n,c)$ is the time
required to pierce pairwise intersecting $c$-oriented polygons.  

In this note, we are interested in studying the performance of the greedy
algorithm for computing the minimum clique cover (MCC) and maximum independent
set (MIS) of a set $G({\cal R})$ of axis-parallel rectangles in 2D. The study on
MCC and MIS problems on a random instance of interval graph have started long
back by Scheinerman \cite{Scheinerman80}. They showed a tight bound of $\sqrt{n}$ 
for both MCC and MIS problems on randomly generated intervals. 
Recently, Tran \cite{Tran11} formally proved that size of the
MCC for $G({\cal R})$ can be at most $O(\sqrt{n} \log \log n)$. He conjectured
that size of the MCC is at most $O(\sqrt{n})$ for $G({\cal R})$. He also
provided results for higher dimension of the problem. A lot of studies on MCC
and MIS problems have been done for the intersection graph of a set of
axis-parallel rectangles on a 2D plane. Aronov et al. \cite{AES09} proposed a
randomized polynomial time algorithm for the MCC problem using the concept of
$\epsilon$-net that can produce $O(\log \log n)$ factor approximation result.
Later Pach et al. \cite{PT11} showed that this is the best possible
approximation for the MCC problem that can be achieved in polynomial time. 
For the MIS problem, the first approximation algorithm for arbitrary
sized axis-parallel rectangles was proposed by Agarwal  et al. \cite{ag} that
produces a $O(\log n)$-factor approximation result in $O(n\log n)$ time. The 
best known result for MIS problem is due to Chalermsook and Chuzhoy 
\cite{ChalermsookC09}, which provides a polynomial time $O(\log \log n)$ factor 
approximation algorithm for the said problem. A nice literature on MIS problem 
can be found in \cite{ChalermsookC09}. Although there is a lower-bound proof on
approximation ratio of MCC problem, same for MIS problem is not known. Nielson
\cite{Neilson00} gives a construction which gives $\Omega(\log n)$ bound for the
greedy algorithm posed in \cite{Chv79} (see Figure \ref{fig}).

Halld{\'o}rsson and Radhakrishnan \cite{HJ94} showed that for general graphs of
bounded degree $\Delta$, the greedy algorithm produces
$\frac{\Delta+1}{3}$-factor approximation result
for the MIS problem. It also proposes a simple parallel algorithm that runs in
$O(\log^*n)$ time using linear number of processors. The approximation factor
can be improved to $\frac{2\overline{d}+3}{5}$ using the fractional relaxation
technique of Nemhauser and Trotter \cite{NT75}, where $\overline{d}$ is the
average degree of the graph. Finally, it shows that using the greedy
strategy of removing all cliques of same size gradually improves the
approximation ratio of the algorithm, and $\frac{\Delta}{3.76}$ approximation
factor is possible to achieve.  Mestre \cite{Mestre06} introduced the notion of
$k$-extendible systems and showed
that for such a system the greedy algorithm produces a $\frac{1}{k}$-factor
approximation result. They showed that maximum weight $b$-matching forms a
2-extendible system, and hence greedy algorithm produces a $\frac{1}{2}$ factor
approximation result. Several other problems, namely maximum profit scheduling,
maximum asymmetric TSP, can be shown to satisfy the properties of $k$-extendible
system for some suitable $k$.

\begin{figure}[t]
\centering
\includegraphics[scale=0.3]{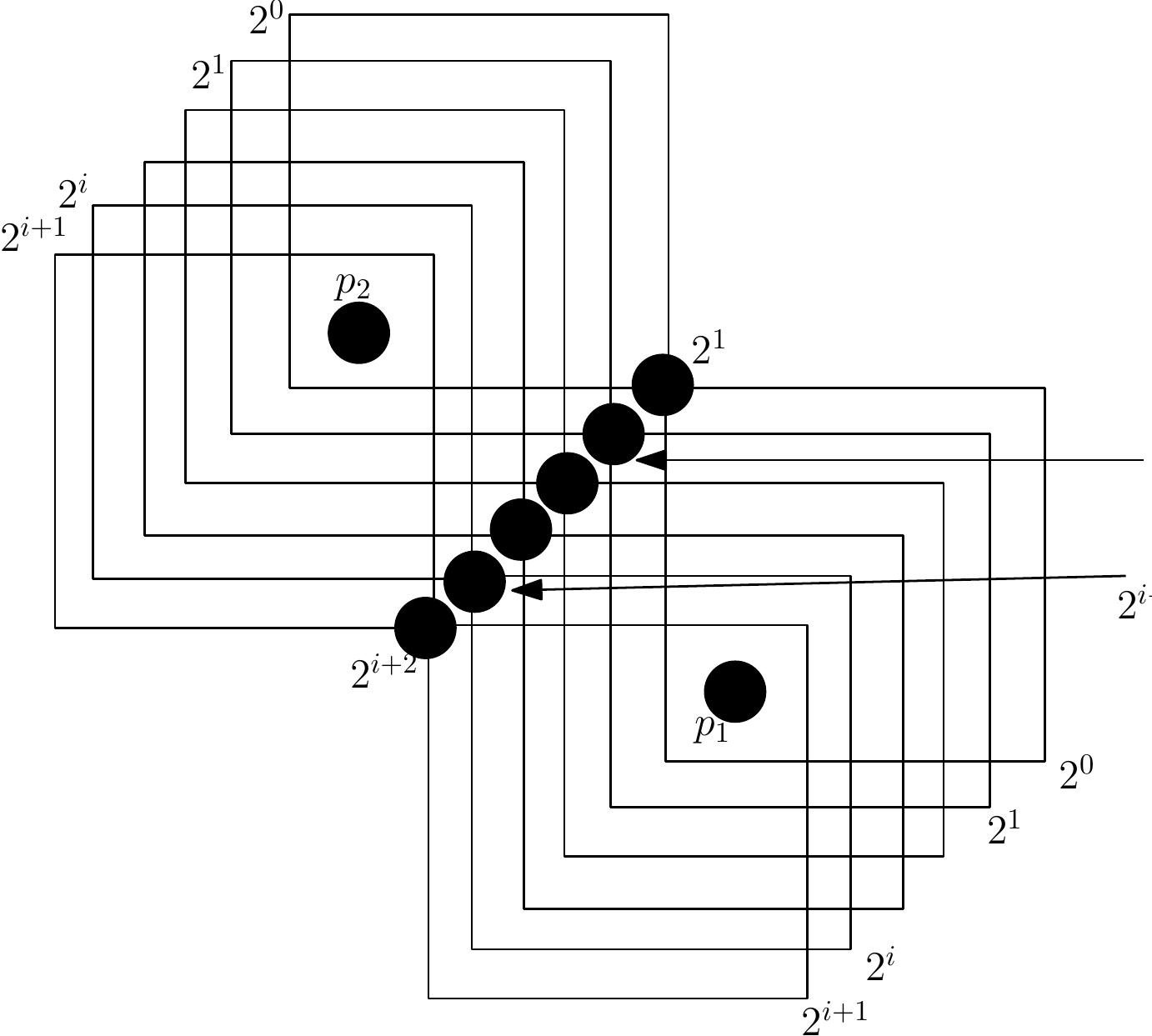} 
\caption{An example showing that  approximation factor of the greedy
algorithm is $\log n$}
\label{fig}
\vspace{-0.1in}
\end{figure}

\vspace{-0.1in}
\section{Our contribution}
\vspace{-0.1in}
The main contribution of this paper is an empirical study showing that the
greedy algorithms for each of the MCC and MIS problems produces very close 
result when compared with optimum solution for a random instance of the
intersection graph  of axis-parallel rectangles. Our experiment pushes the
conjecture made by Tran \cite{Tran11} more towards affirmative side because the
size of the clique cover produced by the greedy algorithm for the MCC problem on
a random instance of rectangle intersection graph is at most $3 \sqrt{n}$ for
reasonably big random instances and we believe that it will hold true
for any such instances. Similarly, the greedy algorithm for the MIS problem also
produces a very close solution to the optimum solution in the sense that the
size of MIS produced by greedy algorithm is at least $2 \times \sqrt{n}$.
Finally we will provide two refined greedy algorithms and their characteristics
that may be of interest to design constant factor approximation algorithms for
the MCC and MIS problems where the participating rectangles are randomly placed
and of random size. We strongly believe that the result also holds for the 
intersection graph of randomly generated axis-parallel rectangles in any finite
dimension. We have also produced the result of Nielson's \cite{Neilson00} divide
and conquer algorithm for the MCC problem. Nielson \cite{Neilson00} has done
some experimental study of his algorithm on the intersection graph of randomly
generated rectangles and claimed that the approximation ratio of their algorithm
is $3.42$. We show that the size of the clique cover produce by our greedy
algorithm is almost 2 times better than that of \cite{Neilson00}.

\section{Algorithms} \label{alg}
We generate a set ${\cal R}=\{R_1,R_2,\ldots,R_n\}$ of axis-parallel rectangles
in a given rectangular region $([a,b],[c,d])$. Each rectangle is stored in the
form of a pair of points indicating its bottom-left and top-right corners. In
other words, for the rectangle $R_i$, we randomly choose a pair of points
$p_i=(x_p,y_p)$, $q_i=(x_q,y_q)$ on the given rectangular region
$([a,b],[c,d])$. Without loss of generality assume that $x_p <x_q$. Now, if $y_p
< y_q$, then $R_i=[(x_p,y_p),(x_q,y_q)]$; otherwise $R_i=[(x_p,y_q),(x_q,y_p)]$.
Now, we introduce the concept of {\it dominated rectangle}. If a rectangle $R_i$
contains some other rectangle $R_j$, $j\neq i$, then it is called a dominated
rectangle. A dominated rectangle can be disregarded since it need not be
considered while computing the minimum clique cover, and also it does not
participate in the maximum independent set. After the generation of $n$ random
rectangles, let $\hat{\cal R}$ denote the set of all non-dominated rectangles.

We use $G(\hat{\cal R})$ to denote the intersection graph of the members in 
$\hat{\cal R}$. Here the nodes correspond to the members in $\hat{\cal R}$.
Between a pair of vertices there is an edge if the corresponding two rectangles
share a common point in their interior. We use $N(R_i)$ to denote the vertices
adjacent to $R_i$ in $G(\hat{\cal R})$ including itself. We now describe the two
greedy heuristic algorithms for the MCC and MIS problems. We have also
implemented the divide and conquer algorithm of Nielson \cite{Neilson00} for the
MCC problem to justify the approximation bound of our algorithms on the MCC
problem for a random instant of rectangle intersection graph. We use this
approximation bound to justify the approximation bound of our greedy algorithm
for the MIS problem for a random instant of rectangle intersection graph. 

\subsection{MCC problem}  \label{GCC}
Since the axis-parallel rectangles satisfy Helly property, the members of a
clique in $G({\cal R})$ share a point in their interior. Thus, the minimum
clique cover of the graph $G({\cal R})$ is same as the minimum number of points
required to stab all the rectangles in $\hat{\cal R}$. 
\subsubsection{Greedy algorithm}
Our greedy algorithm proceeds as follows. We perform a horizontal line sweep
from top to bottom to compute the largest clique ${\cal C}$ of $G({\cal R})$. It
is a point where maximum number of rectangles in $\hat{\cal R}$ overlap. All
these rectangles can be stabbed by a single point. We delete all the rectangles
in ${\cal C}$ from $\hat{\cal R}$ and repeat the same process. The iteration
continues until all the rectangles in $\hat{\cal R}$ are deleted. The pseudo
code of the algorithm is given below. 

\begin{algorithm}[t!]
\caption{Algorithm GCC$({\cal R})$}
\small
\begin{algorithmic}[1]
\STATE {\bf Input:} A set $\cal R$ of randomly generated $n$ rectangles on a
2-dimensional plane. 
\STATE {\bf Output:} A set $GCC$ of points that stab all the rectangles in 
$\cal R$.
\STATE Compute $\hat{{\cal R}}$, the non-dominated set of rectangles in $\cal
R$.
\STATE Initialize $GCC =\emptyset$.
\REPEAT
\STATE Compute a maximum clique ${\cal C}$ among the rectangles in
$\hat{\cal R}$ using the algorithm of \cite{IA83,NB}, and choose a point $\pi$
in the common region of the members in $\cal C$.  
\STATE Set $\hat{\cal R}= \hat{\cal R} \setminus {\cal C}$ and $GCC= GCC \cup
\{\pi\}$.
\UNTIL $\hat{\cal R} =\emptyset$.
\STATE Return $GCC$ as the stabbing points along with its cardinality.
\end{algorithmic}
\label{alg:GCC}
\normalsize
\end{algorithm}
\vspace{-0.1in}
\begin{lemma}
The worst case time complexity of the algorithm $GCC$ is $O(n^2\log n)$.  
\end{lemma}
\begin{proof}
Follows from the fact that  the largest clique of a
rectangle intersection graph can be computed in $O(n\log n)$ time
\cite{IA83,NB}, and the number of iterations can be $O(n)$ in the worst case. 
\end{proof}

\subsubsection{An improved greedy algorithm}
We now introduce the concept of {\it simplicial rectangle} to present an
improvement of the greedy algorithm $GCC\_I$ for the MCC problem. A node in the
graph $G({\cal R})$ is said to be simplicial if all the nodes adjacent to it
form a clique. The corresponding rectangle will be referred to as {\it
simplicial rectangle}. The concept of simplicial rectangle is very much similar
to a simplicial vertex in a graph \cite{Golumbic}. It needs to be noted that all
the rectangles adjacent to a simplicial rectangle can be stabbed by a point
along with the rectangle $R$. Algorithm \ref{alg:GCCI} states the detailed
procedure. 
 
\begin{algorithm}[t!]
\caption{Algorithm GCC\_I$({\cal R})$}
\small
\begin{algorithmic}[1]
\STATE {\bf Input:} A set $\cal R$ of $n$ randomly generated rectangles on a
2D plane.
\STATE {\bf Output:} A set $GCC\_I$ of points that stab all the rectangles in 
$\cal R$.
\STATE Compute the non-dominated set of rectangles $\hat{\cal R}$.
\STATE Initialize $\Theta =\emptyset$; $\Phi=\emptyset$.
\REPEAT
\STATE $R$ = Find\_Simplicial($\hat{\cal R}$).
\IF{$R \neq 0$}
\STATE Let $N(R)$ be the rectangles adjacent to $R$ including itself. 
\STATE Choose a point $\pi$ in the region common to all the members in
$N(R)$. 
\STATE Set $\hat{\cal R}= \hat{\cal R} \setminus N(R)$ and $\Theta = \Theta
\cup \{\pi\}$.
\ELSE
\STATE Compute a maximum clique ${\cal C}$ among the rectangles in
$\hat{\cal R}$ using the algorithm of \cite{IA83,NB}, and choose a point $\pi$
in the common region of the members in $\cal C$.  
\STATE Set $\hat{\cal R}= \hat{\cal R} \setminus {\cal C}$ and $\Phi = \Phi
\cup \{\pi\}$.
\ENDIF
\UNTIL {$\hat{\cal R} =\emptyset$}
\STATE {\bf return} $GCC\_I=\Theta \cup \Phi$ as the stabbing points along with
its 
cardinality.
\end{algorithmic}
\vspace{0.2in}

\begin{algorithmic}[1]
\STATE {\bf Procedure Find\_Simplicial}($\mathbf{A}$)
\STATE {\bf Assumption:} All the vertices in $A$ are unmarked. 
\STATE Compute an incidence matrix $I$ of the graph $G(A)$; 
\STATE Sort the vertices of $G(A)$ in increasing order of their degrees;
\WHILE{all the vertices in $A$ are not marked}
\STATE choose the vertex $v$ having minimum degree among the unmarked vertices. \\
Let $N(v)$ be the set of vertices adjacent to $v$. $N(v)$ includes the vertex
$v$. 
\IF{$N(v)$ forms a clique}
\STATE {\bf return} $R$ as a simplicial rectangle, and a point $\pi$ 
inside the common intersection region;
\ELSE 
\STATE Mark the members of $N(v)$ since they can not be simplicial.
\IF{$a,b \in N(v)$ and $I[a,b] = 0$} 
\STATE mark all the vertices $u$ such that $I[a,u]=1$ and $I[b,u]=1$
\ENDIF
\ENDIF
\ENDWHILE
\STATE {\bf return} 0.
\end{algorithmic}
\label{alg:GCCI}
\normalsize
\end{algorithm}
\vspace{-0.1in}
\begin{lemma} \label{l2}
The worst case time complexity of the algorithm GCC\_I is $O(n^3)$.
\end{lemma}
\begin{proof}
Each iteration of the repeat loop of the algorithm consists of two steps: (i) 
finding a simplicial rectangle, and (ii) finding the largest clique. 

Let $N(R)$ denote the neighbors of $R$ including itself in the graph $G({\cal
R})$. In Step (i), the construction of the incidence matrix $I$ needs $O(n^2)$
time. While searching for a simplicial rectangle, in each failure step it marks
all the members of $N(R)$ since there exists no other rectangle $R' \in N(R)$
which is simplicial. The proof is as follows. If $N(R)=N(R')$ then surely $R'$
is not simplicial. Otherwise there exists some rectangle $R''$ such that $R''
\in N(R')$ but $R'' \not\in N(R)$; since $R \in N(R')$, $R'$ is not
simplicial. Moreover, if $\rho,\rho'\in N(R)$ and $I(\rho,\rho')=0$, then we
have deleted all the rectangles having both $\rho$ and $\rho'$ as neighbors.
Thus, each entry of the matrix $I$ is accessed $O(1)$ time. 

Step  (ii) needs $O(n\log n)$ time in the worst case \cite{IA83,NB}. Since the 
number of iterations is $O(n)$ in the worst case, the result follows. 
\end{proof}
\vspace{-0.1in}
\subsection{MIS problem}  \label{MIS}
\subsubsection{Greedy algorithm} \label{MISG}
Our MIS heuristic also depends on the concept of simplicial rectangle. If a 
simplicial rectangle $R$ is found in $G(\hat{\cal R})$ (i.e., $N(R)$ forms
a clique), we can only choose $R$ in the independent set among the set of
rectangles $N(R)$. Our algorithm is an iterative one. At each iteration, it
searches for a simplicial rectangle. If such a rectangle $R$ is found, it is
included in MIS; otherwise, we delete a rectangle having maximum number of
neighbors. The logic behind choosing such a rectangle is that its absence may
delete a neighbor of maximum number of rectangles. Note that, this may also be a
simplicial rectangle if some of its adjacent rectangle is removed. In Figure
\ref{fig1} such a situation is demonstrated. Here none of the rectangles present
in the region is simplicial due to the position of the other rectangles.
Rectangle $A$ has maximum number of neighbors. But removal of rectangle $B$
makes it simplicial. However, the chance of such a rectangle to be simplicial is
small. The pseudo code of the algorithm is given in Algorithm \ref{alg:MIS}.
\begin{figure}[t]
\centering
\includegraphics[scale=0.4]{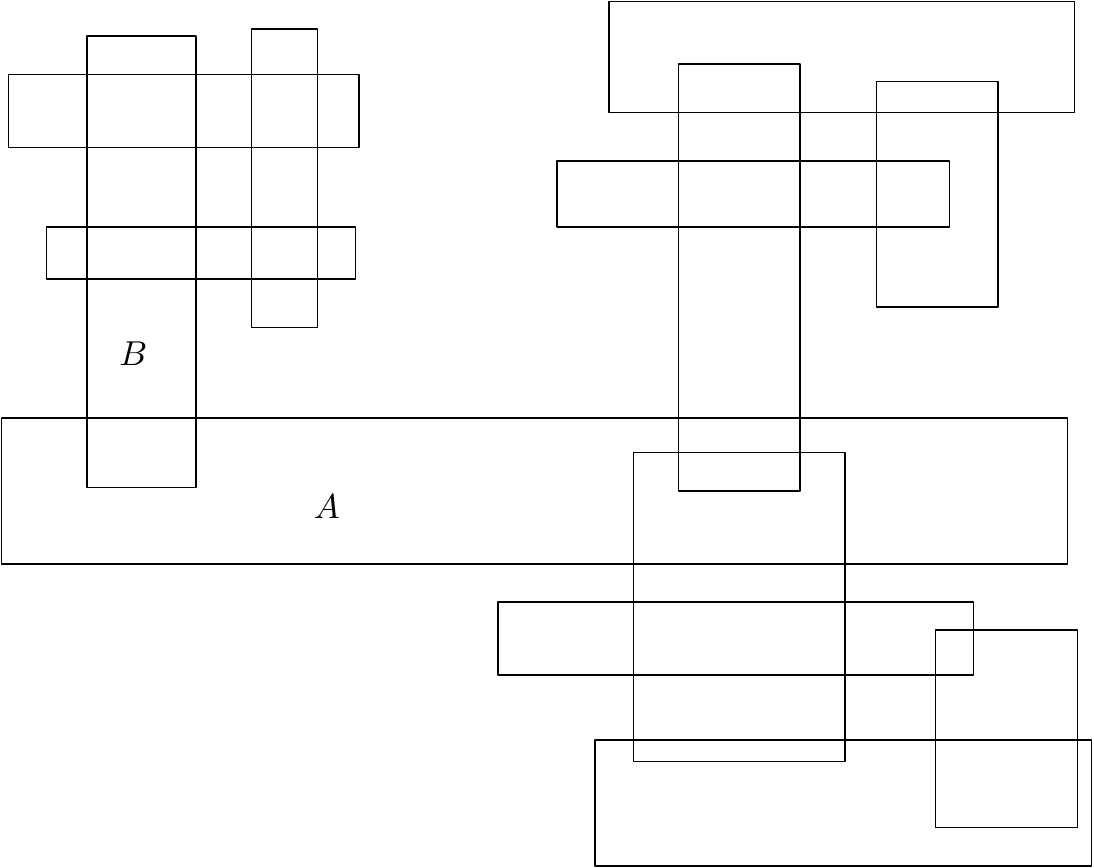} 
\caption{The maximum degree rectangle removal may lead to non-optimal result}
\label{fig1}
\vspace{-0.1in}
\end{figure}

\begin{algorithm}[h!]
\caption{Algorithm MIS$({\cal R})$}
\small
\begin{algorithmic}[1]
\STATE {\bf Input:} A set $\cal R$ of $n$ randomly generated rectangles on a 2D
plane.
\STATE {\bf Output:} A set $MIS$ of mutually non-overlapping rectangles.

\STATE Compute $\hat{\cal R}$, the non-dominated set of rectangles in $\cal R$.
\STATE Initialize $MIS =\emptyset$.
\REPEAT
\STATE $R$ = Find\_Simplicial($\hat{\cal R}$).
\IF{$R \neq 0$}
\STATE Let $N(R)$ be the set of rectangles adjacent to $R$ including itself.  
\STATE Set $\hat{\cal R}= \hat{\cal R} \setminus N(R)$ and $MIS = MIS
\cup \{R\}$.
\ELSE
\STATE Let $R'$ be the rectangle having maximum degree.
\STATE Set $\hat{\cal R}= \hat{\cal R} \setminus \{R'\}$
\ENDIF
\UNTIL {$\hat{\cal R} =\emptyset$}
\STATE Return $MIS$ along with its cardinality.
\end{algorithmic}
\label{alg:MIS}
\normalsize
\end{algorithm}
\vspace{-0.1in}
\begin{lemma}
The worst case time complexity of the algorithm MIS is $O(n^3)$.
\end{lemma}
\vspace{-0.1in}
\begin{proof}
Each iteration of the repeat loop of the algorithm consists of two steps: 
(i) finding a simplicial rectangle, and (ii) finding the rectangle having
maximum degree.  The result follows from the time complexity of Step (i) (see
Lemma \ref{l2}). 
\end{proof}
\vspace{-0.1in}
\subsubsection{A variation of the Greedy algorithm}
In our greedy algorithm (subsection \ref{MISG}), at each iteration we searched
for a simplicial rectangle. If such a rectangle is not found, then we removed
the rectangle having maximum degree, and repeated the process. 

The algorithm proposed in this section is very similar to $GCC\_I$ algorithm.
Here instead of removing the rectangle having maximum degree, we identified the
largest clique and then removed all the rectangles participating in that clique.
Finally, we report the set simplicial rectangles as the independent set.
The time complexity of this algorithm remains same as that of $GCC\_I$.

\vspace{-0.1in}
\section{Experimental Studies}
\vspace{-0.1in}
We have performed a detailed experiment with different  $n$ (number of
rectangles). For each $n$, we generated 20 different instances of $n$ random
rectangles as described in Section \ref{alg}. For each instance, we run our
proposed heuristics for both MCC and MIS problem, and also the divide and
conquer heuristics of \cite{Neilson00} for the MCC problem. In Table
\ref{T} and Figure \ref{fig:GCC}, we refer this algorithm as DCC. Figure
\ref{fig:GCC} shows the comparison of performance of our proposed two greedy
heuristics $GCC$ and $GCC\_I$ and the divide and conquer algorithm DCC of
\cite{Neilson00} for the MCC problem on rectangle intersection graph. It is
observed that our both the algorithms produce result better that of
\cite{Neilson00}. We have also plotted $3\sqrt{n}$ for different values of $n$
in the same graph to demonstrate the solution produced by the improved greedy
heuristic $GCC\_I$ is always less than $3\sqrt{n}$. 

\begin{figure}[ht]
\centering
\includegraphics[scale=0.375]{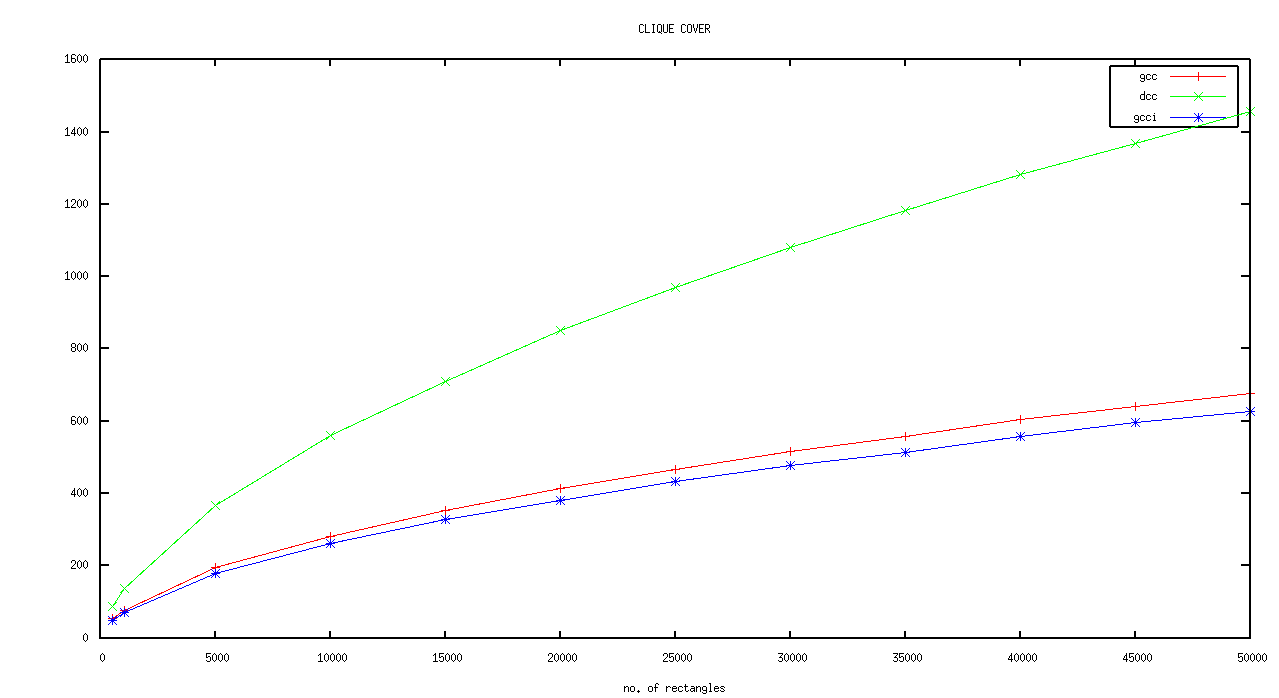} 
\caption{Comparison of the quality of solution produced by our proposed 
heuristics $GCC$, $GCC\_I$ and the DCC algorithm of 
\cite{Neilson00} for the clique cover
of  rectangle intersection graph}\label{fig:GCC}
\end{figure}

\begin{figure}[htbp]
    \centering
    \includegraphics[scale=0.6]{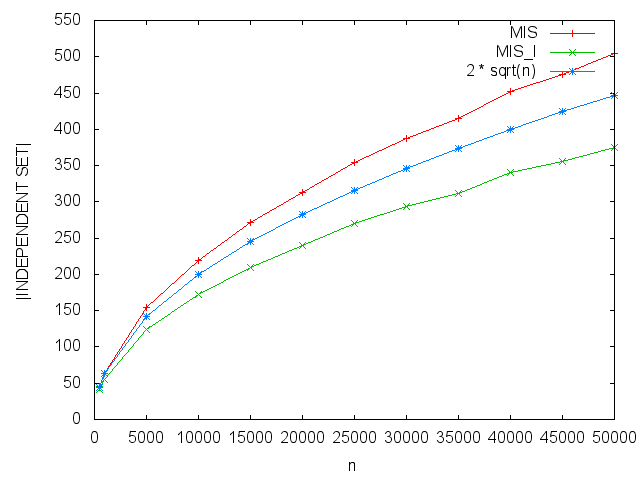}\vspace{-0.1in}
\caption{Comparison of the results produced by our proposed heuristics
$MIS$ and $MIS\_I$ for the maximum independent set of a rectangle
intersection graph}
\label{fig:MIS}
\end{figure}

In Figure \ref{fig:MIS}, we demonstrate the performance of our proposed two
greedy heuristics $MIS$ and $MIS\_I$ for the maximum independent set problem on
rectangle intersection graph. We have also plotted $2\sqrt{n}$ for different
values of $n$ in the same graph to demonstrate the solution produced by the
greedy heuristic $MIS$ is always less than $2\sqrt{n}$. But, the running time of
our $MIS\_I$ heuristic is much better than that of $MIS$.

\begin{table*}
\caption{Experimental result}
\small
\begin{center}
\begin{tabular}{|c|c|c|c|c|c|c|}\hline
$n$ & \multicolumn{3}{c|}{min. clique cover}  & \multicolumn{2}{c|}{max. indep.
set}  & \\  \cline{2-6}  
& DCC\cite{Neilson00} & GCC & GCC\_I & MIS & MIS\_I &  $GCC\_I/MIS$ \\  \hline
500     & 87   & 53  & 48  & 46  & 41  & 1.0434\\
1000    & 137  & 77  & 71  & 64  & 55  & 1.1094\\
5000    & 367  & 195 & 180 & 155 & 124  & 1.1613\\
10000  & 560  & 281 & 262 & 219 & 172  & 1.1963\\
15000  & 709  & 353 & 327 & 271 & 210  & 1.2066\\
20000  & 851  & 414 & 382 & 313 & 240  & 1.2204\\
25000  & 971  & 467 & 432 & 354 & 270  & 1.2203\\
30000  & 1080 & 515 & 478 & 388 & 294  & 1.2320\\
35000  & 1183 & 557 & 514 & 415 & 311  & 1.2386\\
40000  & 1282 & 604 & 559 & 452 & 341  & 1.2367\\
45000  & 1369 & 642 & 596 & 475 & 356  & 1.2547\\
50000  & 1456 & 677 & 628 & 504 & 375  & 1.2460\\ \hline
\end{tabular}
\end{center}
\normalsize
\label{T}
\vspace{-0.2in}
\end{table*}

The final conclusion of our experimental study is summarized in Observation
\ref{ob-final}, and is demonstrated in Figure \ref{fig:approx}, The
justifications of getting such results are also explained.

\begin{observation} \label{ob-final}
The solution produced by our greedy heuristics for the minimum clique cover
(MCC) problem on an intersection graph of a set of randomly generated
axis-parallel rectangles is  at most $2 \times OPT_{MCC}$, where $OPT_{MCC}$ is
the size of the optimum solution of the same problem. 
\end{observation}

\begin{figure}[ht]
\centering
\includegraphics[scale=0.6]{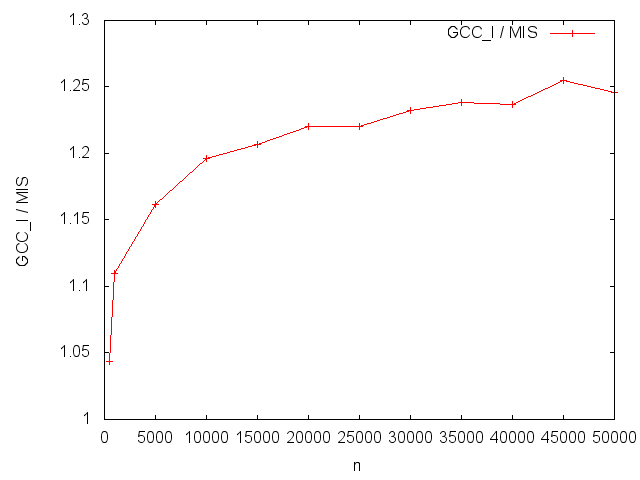} 
\caption{Justification of Observation \ref{ob-final}} 
\label{fig:approx}
\end{figure}

{\bf Justification 1:} Let $OPT_{MIS}$ be the size of the optimum solution of
the maximum independent set (MIS) problem. We have $\frac{OPT_{MIS}}{OPT_{MCC}}
\leq 1$. Figure \ref{fig:approx} shows that $\frac{|GCC\_I|}{|MIS|} \leq
1.5 < 2$ for all the chosen values  of $n$, where $|GCC\_I|$ and $|MIS|$
are the size of the solution generated by our greedy heuristics $GCC\_I$ and
$MIS$ for the MCC and MIS problems respectively. Again $|OPT_{MCC}| \leq |CC|$
for any arbitrary clique cover $CC$ and $|OPT_{MIS}| \geq |IS|$ for any
arbitrary independent set $IS$ of the given graph. Thus we have \\

\centerline{$OPT_{MIS} \leq OPT_{MCC} < |CC| < 2 \times |IS| \leq 2
\times OPT_{MIS}$.} 

\vspace{0.1in}
{\bf Justification 2:} In our experiment it is observed that 
$\frac{|MIS\_I|}{|GCC\_I|} < 1$ for all the values of $n$ we have chosen.
During the execution of $GCC\_I$ we observed two parameters $\Phi$ and $\Psi$,
where $\Phi$ indicates the number of simplicial rectangles observed, and $\Psi$
indicates that the number of times we need to eliminate the largest clique
without getting a simplicial rectangle. It is sure that the simplicial
rectangles need to be stabbed by a point; so $\Phi$ many points are essential.
$\Psi$ is the set of extra points used to stab the rectangles that are not 
stabbed by any point in $\Phi$, and $\Psi$ is less than $\Phi$. So, this gives
an indication towards 2 factor approximation algorithm for the MIS problem on
the randomly generated rectangle intersection graph. The indication could be
well justified if after the elimination of each clique corresponding to a
member in $\Psi$, we could get a simplicial rectangle. But the execution trace
(which is not included in this note) does not demonstrate this fact.\\

{\bf Justification 3:} It is noticed that since the solution produced by both 
GCC and GCC\_I algorithms for the clique cover problem is very close to 
$2\sqrt{n}$. It is proved in \cite{Neilson00} that the size of the optimum 
solution of the MCC problem is upper bounded by $\sqrt{n}$ for a set of $n$ 
randomly positioned rectangles. Thus, the empirical evidences show that our 
algorithm produces 2 approximation result for large values of $n$. 

\vspace{-0.1in}
\section{Conclusion}
\vspace{-0.1in}
In this note, we experimentally analyze the performance of greedy algorithm for
the minimum clique cover and maximum independent sets problems for rectangle
intersection graphs. Exerimental result shows that it produces 1.5 factor
approximation on the randomly generated instances of the corresponding problems.
The intuitive justifications of such behavior may lead to a formal algorithm of
getting a constant factor approximation results of the corresponding problems
for random instances of rectangle intersection graph. 

\small
\bibliographystyle{abbrv}
\bibliography{BibTex}

\end{document}